\documentclass[letterpaper,11pt]{article}
\usepackage{fullpage}
\usepackage[numbers]{natbib}

\usepackage{amsmath}
\usepackage{amsfonts}
\usepackage{amssymb}
\usepackage{amsthm}
\usepackage[noend]{algorithmic}
\usepackage{algorithm}
\usepackage{color}

\def\draft{1}
\def\submit{1}

\ifnum\draft=1 
    \def\ShowAuthNotes{1}
\else
    \def\ShowAuthNotes{0}
\fi

\ifnum\ShowAuthNotes=1
\newcommand{\authnote}[2]{{ \footnotesize \bf{\color{red}[#1's Note: {\color{blue}#2}]}}}
\else
\newcommand{\authnote}[2]{}
\fi

\ifnum\submit=0 
\newcommand{\forsubmit}[1]{#1}
\newcommand{\forreals}[1]{}
\else
\newcommand{\forreals}[1]{#1}
\newcommand{\forsubmit}[1]{}
\fi

\ifnum\draft=1 
\else
\fi

\newtheorem{theorem}{Theorem}[section]

\newtheorem{remark}{Remark}[section]
\newtheorem{lemma}[theorem]{Lemma}

\newtheorem{fact}{Fact}[section]
\newtheorem{observation}{Observation}[section]

\theoremstyle{definition}
\newtheorem{definition}{Definition}[section]

\newcommand{\Esymb}{\mathbb{E}}
\newcommand{\Psymb}{\mathbb{P}}

\DeclareMathOperator*{\E}{\Esymb}

\DeclareMathOperator*{\ProbOp}{\Psymb r}
\renewcommand{\Pr}{\ProbOp}

\renewcommand{\hat}{\widehat}

\newcommand{\R}{\mathbb{R}}

\newcommand{\remove}[1]{}
\newcommand{\eps}{\varepsilon}
\renewcommand{\epsilon}{\varepsilon}

\newcommand{\bR}{\mathbb{R}}

\begin{document}


\title{Take it or Leave it: Running a Survey when Privacy Comes at a Cost}
\author{Katrina Ligett\thanks{Computing and Mathematical Sciences and Division of the Humanities and Social Sciences, Caltech. Email: {\tt katrina@caltech.edu}}\and
Aaron Roth\thanks{Department of Computer and Information Sciences, University of Pennsylvania. Email: {\tt aaroth@cis.upenn.edu}}}

\maketitle

\begin{abstract}
  In this paper, we consider the problem of estimating a potentially sensitive
  (individually stigmatizing) statistic on a
  population. In our model, individuals are concerned about their privacy, and
  experience some \emph{cost} as a function of their privacy
  loss. Nevertheless, they would be willing to participate in the
  survey if they were compensated for their privacy cost.
 These cost functions are not publicly known, however, nor do we make
 Bayesian assumptions about their form or distribution.
Individuals
  are rational and will misreport their costs for privacy if doing so
  is in their best interest. Ghosh and Roth recently showed in
  this setting, when costs for privacy loss may be correlated with
  private types, if individuals value \emph{differential privacy},
  no individually rational direct revelation mechanism can compute
  any non-trivial estimate of the population statistic. In this paper,
  we circumvent this impossibility result by proposing a modified
  notion of how individuals experience cost as a function of their
  privacy loss, and by giving a mechanism which does not operate by
  direct revelation. Instead, our mechanism has the ability to randomly approach
  individuals from a population and offer them a take-it-or-leave-it
  offer. This is intended to model the abilities of a surveyor who may
  stand on a street corner and approach passers-by.
\end{abstract}


\section{Introduction}
Suppose you are a researcher and you would like to collect data from a
population and perform an analysis on it.  Presumably, you would like your
sample, or at least your analysis, to be representative of the
underlying population.  Unfortunately, individuals' decisions
of whether to participate in your study may skew your data: 
perhaps people with an embarrassing
medical condition are less likely to respond to a survey whose results
might reveal their condition.

One could try to incentivize participation by offering a reward for
participation, but this only serves to skew the survey in favor of
those who value the reward over the costs of participating (e.g.,
hassle, time, detrimental effects of what the study might reveal),
which again may not result in a representative sample.  Ideally, you
would like to be able to find out exactly how much you would have
to pay each individual to participate in your survey (her ``value'', akin to a reservation price),
and offer her exactly that much.  Unfortunately,
traditional mechanisms for eliciting player values truthfully are not
a good match for this setting because a player's value may be correlated
with her private information (for example, individuals with an
embarrassing medical condition might want to be paid extra in order to
reveal it). Standard mechanisms based on the revelation principle are therefore no longer
truthful.  In fact, Ghosh and Roth \cite{GR11} showed that when
participation costs can be arbitrarily correlated with private data,
no direct revelation mechanism can simultaneously offer non-trivial
accuracy guarantees and be individually rational for agents who value
their privacy.

Voluntarily provided data is a cornerstone of medical studies, opinion
polls, human subjects research, and marketing studies.
Some data collectors, such as the US Census, can get around the issue
of voluntary participation by legal mandate, but this is rare.  How
might we still get analyses that represent the underlying population?

Statisticians and econometricians have of course attempted to address
selection and non-response bias issues.  One approach is to assume that
the effect of unobserved variables has mean zero.  The
Nobel-prize-winning Heckman correction method \cite{Heckman} and the
related literature instead attempt to correct for non-random samples
by formulating a theory for the probabilities of the unobserved
variables and using the theorized distribution to extrapolate a
corrected sample. The limitations of these approaches is precisely in
the assumptions they make on the structure of the data. Is it possible
to address these issues without needing to ``correct'' the observed
sample, while simultaneously minimizing the cost of running the
survey?

\subsection{Contributions}
The present paper provides a new model for
incentivizing participation in data analyses when the
subjects' value for their private information may be correlated with
the sensitive information itself.  In this model, we present a
mechanism for eliciting responses that allows us to compute accurate
statistical estimates, addressing the survey problem described above. We model costs for individual's participation using the tools and language of differential privacy; our mechanisms are not specific to user costs defined in terms of differential privacy, but offering guarantees of this type can significantly decrease costs when compared to mechanisms that ask for unrestricted access to user data.

Of course a second issue beyond representative participation is
\emph{truthful} participation. We require that rational agents be positively incentivized to participate in our mechanism, but once we get their participation, there is the question of whether they will answer the survey question correctly. One solution is to assume that survey
responses are verifiable or cannot easily be fabricated (e.g., the surveyor requires
documentation of answers, or, more invasively, actually collects a
blood sample from the participant).  While the approach we present in
this paper works well with such verifiable responses, in addition, our
framework provides a formal ``almost-truthfulness'' guarantee, that
the expected utility a participant could gain by lying in the survey is at most very
small. Note that this is a different issue than participation, which is voluntary and always strictly incentivized.

\subsection{Differential Privacy}
Over the past decade, differential privacy has emerged as a compelling
privacy definition, and has received considerable attention. While we
provide formal definitions in Section~\ref{sec:prelims}, differential privacy
essentially bounds the sensitivity of an algorithm's output to arbitrary changes
in individual's data. In particular, it requires that the probability
of \emph{any} possible outcome of a computation be insensitive to the addition or removal of one person's data from the
input. Among differential privacy's many strengths are (1) that
differentially private computations are approximately
truthful~\cite{MT07} (which gives the almost-truthfulness guarantee
mentioned above), and (2) that differential privacy is a property of
the \emph{mechanism} and is independent of the input to the mechanism.

\paragraph{How Individuals Should Value Privacy: The ``Paradox'' of Differential Privacy }
A natural approach taken by past work (e.g., \cite{GR11}) in attempting to model
the cost incurred by participants in some computation on their private
data is to model individuals as experiencing cost as a function of the
\emph{differential privacy} parameter $\epsilon$ associated with the
mechanism using their data.  We argue here, however, that modeling an
individual's cost for privacy loss solely as any function $f(\epsilon)$ of
the privacy parameter $\epsilon$ would lead to unnatural agent
behavior and incentives.

Consider an individual who is approached on a street corner and offered
a deal: she can participate in a survey in exchange for \$100, or she
can decline to participate and walk away. She is given the guarantee
that both her participation decision and her input to the survey (if
she opts to participate) will be treated in an
$\epsilon$-differentially private manner. In the usual language of differential privacy, what does this mean?
Formally, her input to the mechanism will be the tuple containing her
participation decision and her private type. If she decides not to
participate, the mechanism output is not allowed to depend on her private
type,
and switching her participation decision to ``yes'' cannot change the
probability of any outcome by more than a small multiplicative
factor. Similarly, fixing her participation decision as ``yes'', any
change in her stated type can only change the probability of any
outcome by a small multiplicative factor.

How should she respond to this offer? A natural conjecture is that she
would experience a higher privacy cost for participating in the survey
than not (after all, if she does not participate, her private type has
\emph{no} effect on the output of the mechanism -- she need not even have provided it), and that she should
weigh that privacy cost against the payment offered, and make her
decision accordingly.

However, if her privacy cost is solely some function $f(\epsilon)$ of
the privacy parameter of the mechanism, she is actually incentivized
to behave quite differently. Since the privacy parameter $\epsilon$ is
\emph{independent} of her input, her cost $f(\epsilon)$ will be
identical \emph{whether she participates or not}. Indeed, her
participation decision does not affect her privacy cost, and only
affects whether she receives payment or not, and so she will always
opt to participate in exchange for any positive payment. Further, she
experiences the full privacy cost of the mechanism simply by being
asked whether she wishes to participate in the survey, even if her
private data is never used!
%

We view this as problematic and as not modeling the true
decision-making process of individuals: in reality, the full privacy
cost of a survey may not have been experienced by the individual before her
data has been given. Furthermore, individuals are unlikely to accept
arbitrarily low offers for their private data.  One potential route to
addressing this ``paradox'' would be to move away from modeling the
value of privacy solely in terms of an input-independent privacy
guarantee.  This is the approach taken by \cite{CCKMV12}. Instead, we retain the framework of differential privacy,
but introduce a new model for how individuals reason about
the cost of privacy loss. Roughly, we model individuals' costs
as a function of the differential
privacy parameter of the portion of the mechanism they participate in,
and assume they do not experience cost from the parts
of the mechanism that process data that they have not provided (or
that have no dependence on their data). For our application, we
consider mechanisms that operate in two stages: first, they aggregate
participation decisions by making take-it-or-leave-it offers, but do not 
compute on the private types collected from the participating
individuals. The ``output'' of this stage of the mechanism is the
observed behavior of the surveyor: the number of people approached and
the prices offered.\footnote{For clarity, in our analysis, we also include as part of the output of this stage a noisy (privacy-preserving) count of the number of people accepting the highest offer we make.} The second stage of the mechanism takes as input
the reported private types of the individuals who elected to
participate. In our model, individuals who declined to provide their private types
do not experience any cost in this second stage of the mechanism.

\subsection{Related Work}
In recent years, differential privacy, which was introduced in a series of papers \cite{DMNS06, BDMN05}, has emerged as the standard solution concept for privacy in the theoretical computer science literature. There is by now a very large literature on this fascinating topic, which we do not attempt to survey here, instead referring the interested reader to a survey by Dwork \cite{DworkSurvey}.

McSherry and Talwar proposed that differential privacy could itself be used as a \emph{solution concept} in mechanism design \cite{MT07}. They observed that any differentially private mechanism is approximately truthful, while simultaneously having some resilience to collusion. Using differential privacy as a solution concept (as opposed to dominant strategy truthfulness) they gave some improved results in a variety of auction settings. Gupta et al. also used differential privacy as a solution concept in auction design \cite{GLMRT10}.

This literature was recently extended by a series of elegant papers by Nissim, Smorodinsky, and Tennenholtz \cite{NST12}, Xiao \cite{Xia11}, Nissim, Orlandi, and Smorodinsky \cite{NOS12}, and Chen et al. \cite{CCKMV12}. This line of work observes (\cite{NST12, Xia11}) that differential privacy does not lead to exactly truthful mechanisms, and indeed that manipulations might be easy to find, and then seeks to design mechanisms that are exactly truthful even when agents explicitly value privacy (\cite{Xia11, NOS12, CCKMV12}).

Feigenbaum, Jaggard, and Schapira considered (using a different notion of privacy) how the implementation of an auction can affect how many bits of information are leaked about individuals' bids \cite{Feigenbaum}. Specifically, they study to what extent information must be leaked in second price auctions and in the millionaires problem. We consider somewhat orthogonal notions of privacy and implementation that make our results incomparable.

Most related to this paper is an orthogonal direction initiated by
Ghosh and Roth \cite{GR11}. Ghosh and Roth consider the problem of a
data analyst who wishes to buy data from a population for the purposes
of computing an accurate estimate of some population
statistic. Individuals experience cost as a function of their privacy
loss (as measured by differential privacy), and must be incentivized
by a truthful mechanism to report their true costs. In particular,
\cite{GR11} show that if individuals experience disutility as a
function of differential privacy, and if costs for privacy can be
arbitrarily correlated with private types, then no individually
rational direct revelation mechanism can achieve any nontrivial
accuracy. In this paper, we overcome this impossibility result by
abandoning the direct revelation model in favor of a model in which a
surveyor can approach random individuals from the population and offer
them take-it-or-leave-it offers, and by introducing a slightly
different model for how individuals experience cost as a function of
privacy. We note that the conversation on how individuals should
experience costs as a function of privacy is ongoing. Ghosh and Roth
\cite{GR11} suggest that privacy costs should be a function of the
differential privacy parameter of the mechanism; \cite{Xia11} suggest
that such costs should be a function of the mutual information between
the agent type and the mechanism output (such a measure requires a
prior on player types); Nissim, Orlandi, and Smorodinsky suggest
that agent costs should merely be \emph{upper bounded} by a linear
function of the privacy parameter of the mechanism \cite{NOS12}; and Chen et
al. \cite{CCKMV12} suggest that the appropriate measure should be
outcome dependent (although inspired by differential privacy). The
model in this paper adds to this fruitful conversation.

Concurrently with this work, Roth and Schoenebeck \cite{RS12} consider
the problem of deriving Bayesian optimal survey mechanisms for
computing minimum variance unbiased estimators of a population
statistic from individuals who have costs for participating in the
survey. Although the motivation of this work is similar, the results
are orthogonal. In this paper, we take a prior-free approach and model
costs for private access to data using the formalism of differential
privacy. In contrast, \cite{RS12} takes a Bayesian approach, assuming
a known prior over agent costs, and does not attempt to provide any
privacy guarantee, and instead only seeks to pay individuals for their
participation.

Also contemporaneously with this work, Fleischer and Lyu \cite{FL12}
consider the problem of computing a statistic over a population where
privacy costs may be correlated with private types. Their approach is
fundamentally different from ours, however, because they crucially
assume the surveyor has perfect knowledge of the correlation between
types and costs. In the present work, we make no such Bayesian assumptions.

\section{Preliminaries}
\label{sec:prelims}

The approach to modeling privacy that we use
is the by-now-standard model of \emph{differential privacy}.


We think of databases as being an ordered multiset of elements from some
universe $X$: $D \in X^*$ in which each element corresponds to the
data of a different individual. We call two databases \emph{neighbors}
if they differ in the data of only a single individual.
\begin{definition}
Two databases of size $n$ $D, D' \in X^n$ are \emph{neighbors} with respect to individual $i$ if for all $j \neq i \in [n]$, $D_j = D'_j$.
\end{definition}
We can now define \emph{differential privacy}. Intuitively,
differential privacy promises that the output of a mechanism does not
depend too much on any single individual's data.
\begin{definition}[\cite{DMNS06}]
A randomized algorithm $A$ which takes as input a database $D \in X^*$ and outputs an element of some arbitrary range $R$ is $\epsilon_i$-differentially private with respect to individual $i$ if for all databases $D, D' \in X^*$ that are neighbors with respect to individual $i$, and for all subsets of the range $S \subseteq R$, we have:
$$\Pr[A(D) \in S] \leq \exp(\epsilon_i)Pr[A(D') \in S]$$
$A$ is $\epsilon_i$-minimally differentially private with respect to individual $i$ if $\epsilon_i = \inf(\epsilon \geq 0)$ such that $A$ is $\epsilon$-differentially private with respect to individual $i$. When it is clear from context, we will simply write $\epsilon_i$-differentially private to mean $\epsilon_i$-minimally differentially private.
\end{definition}

A simple and useful fact is that \emph{post-processing} does not affect differential privacy guarantees.
\begin{fact}
\label{fact:postprocessing}
Let $A:X^*\rightarrow R$ be a randomized algorithm which is
$\epsilon_i$-differentially private with respect to individual $i$,
and let $f:R\rightarrow T$ be an arbitrary (possibly randomized)
function mapping the range of $A$ to some abstract range $T$. Then the
composition $g\circ f: X^*\rightarrow T$ is
$\epsilon_i$-differentially private with respect to individual $i$.
\end{fact}

A useful distribution is the \emph{Laplace} distribution.
\begin{definition}[The Laplace Distribution]
The Laplace Distribution with mean $0$ and scale $b$ is the
distribution with probability density function: $\textrm{Lap}(x | b)
= \frac{1}{2b}\exp( -\frac{|x|}{b})$.
We will sometimes write $\textrm{Lap}(b)$ to denote the Laplace distribution with scale $b$, and will sometimes abuse notation and write $\textrm{Lap}(b)$ simply to denote a random variable $X \sim \textrm{Lap}(b)$.
\end{definition}
A fundamental result in data privacy is that perturbing low sensitivity queries with Laplace noise preserves $\epsilon$-differential privacy.
\begin{theorem}[\cite{DMNS06}]
\label{thm:laplace-privacy}
Suppose $f:X^*\rightarrow \R^k$ is a function such that for all
adjacent databases $D$ and $D'$, $||f(D) - f(D')||_1
\leq 1$. Then the procedure which on input $D$ releases $f(D) +
(X_1,\ldots,X_k)$, where each $X_i$ is an independent draw from a
$\textrm{Lap}(1/\epsilon)$ distribution, preserves
$\epsilon$-differential privacy.
\end{theorem}

We consider a (possibly infinite) collection of individuals drawn from
some distribution over types $\mathcal{D}$. There exists a finite
collection of private types $T$. Each individual is described by a
private type $t_i \in T$, as well as a nondecreasing cost function
$c_i:\bR_+\rightarrow \bR_+$ that measures her disutility
$c_i(\epsilon_i)$ for having her private type used in a computation
with a guarantee of $\epsilon_i$-differential privacy.

Agents interact with the mechanism as follows. The mechanism will be
endowed with the ability to select an individual $i$ uniformly at
random (without replacement) from $\mathcal{D}$, by making a call to a
\emph{population oracle} $\mathcal{O}_\mathcal{D}$. Once an individual
$i$ has been sampled, the mechanism can present $i$ with a
\emph{take-it-or-leave-it offer}, which is a tuple $(p_i,
\epsilon_i^1, \epsilon_i^2) \in \mathbb{R}_+^3$.  $p_i$ represents an
offered payment, and $\epsilon_i^1$ and $\epsilon_i^2$ represent two
privacy parameters. The agent then makes her participation decision,
which consists of one of two actions: she can \emph{accept} the offer,
or she can \emph{reject} the offer. If she accepts the offer, she
communicates her (verifiable) private type $t_i$ to the auctioneer,
who may use it in a computation which is $\epsilon_i^2$-differentially
private with respect to agent $i$. In exchange she receives payment
$p_i$. If she rejects the offer, she need not communicate her type, and
receives no payment. Moreover, the mechanism guarantees that the bit
representing whether or not agent $i$ accepts the offer is used only
in an $\epsilon_i^1$-differentially private way, regardless of her
participation decision.

\subsection{How Cost Functions Relate to Types}
We model agents as caring only about the privacy of their private type
$t_i$, but they may also experience a cost when information about
their cost function $c_i(\epsilon_i)$ is revealed---because of possible
correlations between costs and types. To capture this phenomenon while
still avoiding making Bayesian assumptions, we take the following
approach.

Implicitly, there is some (possibly randomized) process $f_i$ which
maps a user's private type $t$ to his cost function $c_i = f_i(t)$,
but we make no assumption about the form of this map. This takes a
worst case view -- i.e., we have no prior over individuals' cost
functions. For point of reference, in a Bayesian model, the function
$f$ would represent the user's marginal distribution over costs
conditioned on its type. We make no Bayesian assumptions, but
introduce this function $f$ so as to formalize our model of utility
for privacy, which is crucial to the results we give in this paper.

When an individual $i$ is faced with a take-it-or-leave-it offer, her
type is used in two computations: first, her participation decision
(which may be a function of her type) is used in some computation
$A_1$ which will be $\epsilon_i^1$-differentially private. Then, if
she accepts the offer, she allows her type to be used in some
computation $A_2$ which may be $\epsilon_i^2$-differentially private.

We model individuals as \emph{caring about the
  privacy of their cost function only insofar as it reveals
  information about their private type}. Because their cost function
is determined as a function of their private type, if $P$ is some
predicate over cost functions, if $P(c_i) = P(f_i(t_i))$ is used in a
way that guarantees $\epsilon_i$-differential privacy, then the agent
experiences a privacy loss of some $\epsilon_i' \leq \epsilon_i$
(which corresponds to a disutility of some $c_i(\epsilon_i') \leq
c_i(\epsilon_i)$). We write $g_i(\epsilon_i) = \epsilon_i'$ to denote
this correspondence between a given privacy level and the effective
privacy loss due to use of the cost function at that level of
privacy. For example, if $f_i$ is a deterministic injective mapping,
then $f_i(t_i)$ is as disclosive as $t_i$ and so $g_i(\epsilon_i) =
\epsilon_i$. On the other hand, if $f_i$ produces a distribution
independent of the user's type, then $g_i(\epsilon_i) = 0$ for all
$\epsilon_i$.

We note that the cost function model we describe here can also incorporate other costs of participating in a survey not related to privacy concerns, such as valuing time or disliking talking to strangers. One can fold in such considerations (which contribute constant factors independent of the privacy parameter) without changing our model or the qualitative nature of our results, but such a change might result in nonlinear cost functions, violating a simplifying assumption made in the analysis of our mechanism's cost.

\subsection{Cost Experienced from a Take-It-Or-Leave-It Mechanism}
\begin{definition}
A Private Take-It-Or-Leave-It Mechanism is composed of two algorithms, $A_1$ and $A_2$. $A_1$ makes offer $(p_i, \epsilon_i^1, \epsilon_i^2)$ to individual $i$ and receives a binary \emph{participation decision}. If player $i$ \emph{participates}, she receives a payment of $p_i$ in exchange for her private type $t_i$. $A_1$ performs no computation on $t_i$. The privacy parameter $\epsilon_i^1$ for $A_1$ is computed by viewing the input to $A_1$ to be the vector of participation decisions, and the output to be the number of individuals to whom offers were made, the offers $(p_i, \eps_i^1, \eps_i^2)$, and an $\epsilon_i^1$-differentially private count of the number of players who chose to participate at the highest price we offer.

Following the termination of $A_1$, a separate algorithm $A_2$ computes on the reported private types of these participating individuals and outputs a real number $\hat{s}$. The privacy parameter $\epsilon_i^2$ of $A_2$ is computed by viewing the input to be the private types of the participating agents, and the output as $\hat{s}$.
\end{definition}

We assume that agents have quasilinear utility (cost) functions: given a
payment $p_i$, an agent $i$ who declines a take-it-or-leave-it offer
(and thus receives no payment) and whose participation decision is
used in an $\epsilon_i^1$-differentially private way experiences
utility $u_i = - c_i(g_i(\epsilon_i^1)) \geq -c_i(\epsilon_i^1)$. An
agent who accepts a take-it-or-leave-it offer and receives payment
$p$, whose participation decision is used in an
$\epsilon_i^1$-differentially private way, and whose private type is
subsequently used in an $\epsilon_i^2$-differentially private way
experiences utility $u_i = p_i - c_i(\epsilon_i^2 + g_i(\epsilon_i^1)
\geq p_i - c_i(\epsilon_i^2 + \epsilon_i^1)$.


\begin{remark}
  While this model of costs, including the function $g_i$, might seem
  complex, note that it captures the correct cost model in a number of
  situations. Suppose, for example, that costs have correlation $1$
  with types, and $c_i(\epsilon)= \infty$ if and only if $t_i = 1$,
  otherwise $c_i(\epsilon) \ll p_i$. Then, asking whether an agent
  wishes to accept an offer $(p_i, \epsilon_i, \epsilon_i)$ is
  equivalent to asking whether $t_i = 1$ or not, and those accepting
  the offer are in effect answering this question twice. In this case,
  we have $g_i(\epsilon) = \epsilon$. On the other hand, if types and
  costs are completely uncorrelated, then there is no privacy loss
  associated with responding to a take-it-or-leave-it offer. This is
  captured by setting $g_i(\epsilon) = 0$.
\end{remark}

Note that by accepting an offer, agent $i$ achieves utility at least
$$p_i - c_i(\epsilon_i^2 + g_i(\epsilon_i^1)) \geq p_i - c_i(\epsilon_i^2 + \epsilon_i^1).$$
By rejecting an offer, agent $i$ might achieve negative utility,
bounded by:
$$0 \geq -c_i(g_i(\epsilon_i^1)) \geq -c_i(\epsilon_i^1).$$

Agents wish to maximize their utility, and so the following lemma is immediate:
\begin{lemma}
\label{lem:truthful}
A utility-maximizing agent $i$ will accept a take-it-or-leave-it offer $(p_i, \epsilon_i^1, \epsilon_i^2)$ when $p_i \geq c_i(\epsilon_i^1 + \epsilon_i^2)$
\end{lemma}
\begin{proof}
We simply compare the lower bound on an agent's utility when accepting an offer with an upper bound on an agent's utility when rejecting an offer to find that agent $i$ will always accept when
$$p_i - c_i(\epsilon_i^1 + \epsilon_i^2) \geq 0.$$
\end{proof}
\begin{remark}
  Note that this lemma is tight exactly when agent types are
  uncorrelated with agent costs -- i.e., when $g_i(\epsilon) =
  0$. When agent types are highly correlated with costs, then
  rejecting an offer becomes more costly, and agents may accept
  take-it-or-leave-it offers at lower prices.
\end{remark}
We make no claims about how agents respond to offers $(p_i,
\epsilon_i^1, \epsilon_i^2)$ for which $p_i < c_i(\epsilon_i^1 +
\epsilon_i^2)$. Note that since agents can suffer negative utility
even by rejecting offers, it is possible that they will accept offers
that lead to experiencing negative utility. Thus, in our setting,
take-it-or-leave-it offers are not necessarily \emph{truthful} in the
standard sense. Nevertheless, Lemma \ref{lem:truthful} will provide a
strong enough guarantee for us of \emph{one-sided truthfulness}:
we can guarantee that rational agents will accept all offers that
guarantee them non-negative utility.

Note that our mechanisms will satisfy only a relaxed notion of
\emph{individual rationality}: we have not endowed agents with the
ability to avoid having been given a take-it-or-leave it offer, even
if both options (taking or rejecting) would leave her with negative
utility. Agents who reject take-it-or-leave-it offers can experience
negative utility in our mechanism because their rejection decision is
observed and used in a (differentially private)
computation. Once the take-it-or-leave-it offer has been presented,
agents are free to behave selfishly. We feel that both
of these relaxations (of truthfulness and individual rationality) are
well motivated by real world mechanisms in which surveyors may
approach individuals in public, and crucially, they are necessary in
overcoming the impossibility result in \cite{GR11}.

Most of our analysis holds for arbitrary cost functions $c_i$, but we do a benchmark cost comparison assuming \emph{linear} utility functions of the form $c_i(\epsilon) = v_i\epsilon$, for some quantity $v_i$.
\subsection{Accuracy}
Our mechanism is designed to be used by a data analyst who wishes to
compute some statistic about the private type distribution of the
population. Specifically, the analyst gives the mechanism some
function $Q:T\rightarrow [0,1]$, and wishes to compute $a = \E_{t_i
  \sim \mathcal{D}}[Q(t_i)]$, the average value that $Q$ takes among
the population of agents $\mathcal{D}$. The analyst wishes to
obtain an \emph{accurate} answer, defined as follows:
\begin{definition}
A randomized algorithm, given as input access to a population oracle $\mathcal{O}_\mathcal{D}$ which outputs an estimate $M(\mathcal{O}_\mathcal{D}) = \hat{a}$ of a statistic $a = \E_{t_i \sim \mathcal{D}}[Q(t_i)]$ is $\alpha$-accurate if:
$$\Pr[|\hat{a}-a| > \alpha] < \frac{1}{3}$$
where the probability is taken over the internal randomness of the algorithm and the randomness of the population oracle.
\end{definition}
The constant $\frac{1}{3}$ is arbitrary, and is fixed only for
convenience. It can be replaced with any other constant value without
qualitatively affecting any of the results in this paper.

\subsection{Cost}
We will evaluate the cost incurred by our mechanism using a
bi-criteria benchmark: For a parametrization of our mechanism which
gives accuracy $\alpha$, we will compare our mechanism's 
cost to a benchmark algorithm that has perfect knowledge of each
individual's cost function, but is constrained to make every
individual the same take-it-or-leave-it offer (the same fixed price is
offered to each person in exchange for some fixed $\eps'$-differentially
private computation on her private type) while obtaining
$\alpha/32$ accuracy.\footnote{Note that we have made no attempt to optimize the constant.}
That is, the benchmark mechanism must be ``envy-free'', and may obtain better accuracy than we
do, but only by a constant factor. On the other hand, the benchmark mechanism has several advantages: it has full knowledge of each player's cost, and need not be concerned about sample error. For simplicity, we will state our
benchmark results in terms of individuals with linear cost functions.


\section{Mechanism and Analysis}
\subsection{The Take-It-Or-Leave-It Mechanism}

In this section we describe our mechanism, which we present in
Figure~\ref{Alg:harrass}. It is \emph{not} a direct revelation mechanism,
and instead is based on the ability to present take-it-or-leave-it
offers to uniformly randomly selected individuals from some
population. This is intended to model the scenario in which a surveyor
is able to stand in a public location and ask questions or present
offers to passers by (who are assumed to arrive randomly). Those
passing the surveyor have the freedom to accept or reject the offer
that they are presented, but they cannot avoid having the question
posed to them.

Our mechanism consists of two algorithms. Algorithm $A_1$ is run on samples from the population with privacy guarantee $\epsilon_0$, until it terminates at some final epoch $\hat{j}$; and then algorithm $A_2$ is run on $(\textrm{AcceptedSet}_{\hat{j}}, \textrm{EpochSize}(\hat{j}),\epsilon_0)$.

\begin{algorithm}
\caption{Algorithm $A_1$, the ``Harassment Mechanism''. It is parametrized by an accuracy level $\alpha$, and we view its \emph{input} to be the participation decision of each individual approached with a take-it-or-leave-it offer, and its \emph{observable output} to be the number of individuals approached, the payments offered, and the noisy count of the number of players who accepted the offer in the final epoch.}
\label{Alg:harrass}
\begin{algorithmic}
\STATE \textbf{Let} $\textrm{EpochSize}(j) \leftarrow \frac{100(\log j + 1)}{\alpha^2}$.
\STATE \textbf{Let} $j \leftarrow 1$.
\STATE \textbf{Let} $\epsilon_0 = \alpha$
\WHILE{\textbf{TRUE}}
    \STATE \textbf{Let} $\textrm{AcceptedSet}_j \leftarrow \emptyset$ and $\textrm{NumberAccepted}_j \leftarrow 0$
    and $\textrm{Epoch}_j \leftarrow \emptyset$
    \FOR{$i = 1$ to $\textrm{EpochSize}(j)$}
        \STATE \textbf{Sample} a new individual $x_i$ from $\mathcal{D}$.
        \STATE \textbf{Let} $\textrm{Epoch}_j \leftarrow \textrm{Epoch}_j \cup \{x_i\}$.
        \STATE \textbf{Offer} $x_i$ the take-it-or-leave it offer $(p_j, \epsilon_0, \epsilon_0)$ with $p_j = (1+\eta)^{j}$
        \IF{$i$ accepts}
            \STATE \textbf{Let} $\textrm{AcceptedSet}_j \leftarrow \textrm{AcceptedSet}_j \cup \{x_i\}$ and \\ ~~~~~~~~$\textrm{NumberAccepted}_{j} \leftarrow  \textrm{NumberAccepted}_{j} + 1$.
        \ENDIF
    \ENDFOR
    \STATE \textbf{Let} $\nu_{j} \sim \textrm{Lap}(1/\epsilon_0)$ and $\textrm{NoisyCount}_j = \textrm{NumberAccepted}_j + \nu_{j}$
    \IF{$\textrm{NoisyCount}_j \geq (1-\alpha/8)\textrm{EpochSize}(j)$}
        \STATE \textbf{Call} \textbf{Estimate}($\textrm{AcceptedSet}_j, \textrm{EpochSize}(j), \epsilon_0)$.
    \ELSE
        \STATE \textbf{Let} $j \leftarrow j+1$
    \ENDIF
\ENDWHILE
\end{algorithmic}
\end{algorithm}

\begin{algorithm}
\caption{The Estimation Mechanism. We view its \emph{inputs} to be the private types of each participating individual from the final epoch, and its \emph{output} is a single numeric estimate.}
\textbf{Estimate}$(\textrm{AcceptedSet}, \textrm{EpochSize}, \epsilon)$:
\begin{algorithmic}
\STATE \textbf{Let} $\hat{a} = \sum_{x_i \in \textrm{AcceptedSet}}Q(x_i) + \textrm{Lap}(1/\epsilon)$
\STATE \textbf{Output} $\hat{a}/\textrm{EpochSize}$.
\end{algorithmic}
\end{algorithm}

\subsection{Privacy}
Note that our mechanism offers the same $\epsilon_0$ in every take-it-or-leave-it offer it makes.

\begin{theorem}
\label{thm:privacy1}
The Harassment Mechanism is $\eps_0$-differentially private with respect to the participation decision of each individual approached.
\end{theorem}
\begin{proof}
The observable output of $A_1$ is the total number of people approached, the payments offered, and the noisy count of the number of number of players who accepted the offer in the final epoch. The first two of these are functions only of the choice of the final epoch $j$ at which the algorithm halts. But this decision is made as a function only of the quantity $\textrm{NoisyCount}_j$, which preserves $\epsilon_0$-differential privacy by the properties of the Laplace mechanism, and the fact that the vector
$$(\textrm{NumberAccepted}_1,\ldots,\textrm{NumberAccepted}_j)$$ has sensitivity $1$.
\end{proof}

\begin{theorem}
The Estimation Mechanism is $\eps_0$-differentially private with respect to the participation decision and private type of each individual approached.
\end{theorem}
\begin{proof}
The theorem follows directly from the privacy of the Laplace mechanism.
\end{proof}

Note that these two theorems, together with Lemma \ref{lem:truthful}, imply that each agent will accept its take-it-or-leave-it offer of $(p_j,\eps_0,\eps_0)$ whenever $p_j \geq c_i(2\eps_0)$.

\subsection{Accuracy}
\begin{theorem}
Our overall mechanism, which first runs the Harassment Mechanism and then hands the types of the accepting players from the final epoch to the Estimation Mechanism, is $\alpha$-accurate.
\end{theorem}
\begin{proof}
We need simply control four sources of error, which we do in turn. Suppose that the algorithm halts and outputs an answer computed from epoch $\hat{j}$.

We first consider the effects of sample error, namely, the difference between the statistic among those individuals approached in an epoch (not just those who accepted our offer) and the true value of the statistic in the underlying population.
\begin{lemma}\label{lem:close-to-expectation}
Except with probability at most $1/12$, for each epoch $j \leq \hat{j}$ we have:
$$\left|\frac{1}{\textrm{EpochSize}(j)}\sum_{x_i \in  \textrm{Epoch}_j } Q(x_i) - \E_\mathcal{D}[Q(x)]\right| \leq \frac{\alpha}{4}.$$
\end{lemma}
\begin{proof}
This follows from a Chernoff bound and a union bound. Because the individuals in each epoch are sampled i.i.d. from $\mathcal{D}$, by the additive version of the Chernoff bound, we have for each epoch $j$:
\begin{eqnarray*}
\Pr\left[\left|\frac{1}{\textrm{EpochSize}(j)}\sum_{x_i \in  \textrm{Epoch}_j } Q(x_i) - \E_\mathcal{D}[Q(x)]\right| \geq \frac{\alpha}{4}\right] &\leq& 2\exp\left(-\frac{1}{8}\cdot\textrm{EpochSize}(j)\cdot\alpha^2\right) \\
&=& 2\exp\left(-\frac{100}{8}\log(j+1)\right) \\
&=& \frac{2}{\left(j+1\right)^{100/8}}
\end{eqnarray*}
By a union bound, we now have:
\begin{eqnarray*}
\Pr\left[\max_{j}\left|\frac{1}{\textrm{EpochSize}(j)}\sum_{x_i \in  \textrm{Epoch}_j } Q(x_i) - \E_\mathcal{D}[Q(x)]\right| \geq \frac{\alpha}{4}\right] &\leq& \sum_{j=1}^{\hat{j}} \frac{2}{\left(j+1\right)^{100/8}} \\
&\leq& \sum_{j=1}^{\infty} \frac{2}{\left(j+1\right)^{100/8}} \\
&<& \frac{1}{12}
\end{eqnarray*}
\end{proof}

We next consider the impact of the noise added for the purpose of maintaining differential privacy in the Harassment Mechanism.
\begin{lemma}
Except with probability at most $1/12$, for each epoch $j \leq \hat{j}$ we have:
$$|\nu_j| \leq \frac{\alpha}{8}\cdot \textrm{EpochSize}(j)$$
\end{lemma}
\begin{proof}
By the properties of the Laplace distribution, if random variable $Y \sim \textrm{Lap}(b)$, then: $\Pr[|Y| \geq t\cdot b] = \exp(-t)$. By a union bound, we have:
\begin{eqnarray*}
\Pr\left[\max_{j} |\nu_j| \geq \frac{\alpha}{8}\cdot \textrm{EpochSize}(j)\right] &\leq& \sum_{j=1}^{\infty} \Pr\left[|\nu_j| \geq \frac{\alpha}{8}\cdot \textrm{EpochSize}(j)\right] \\
&=& \sum_{j=1}^{\infty} \exp\left(-\frac{100}{8}\log(j+1)  \right) \\
&=& \sum_{j=1}^{\infty} \left(\frac{1}{j+1}\right)^{100/8} \\
&<& \frac{1}{12}
\end{eqnarray*}
\end{proof}
Note that there is an additional loss of accuracy due to the fact that we target only a $(1 - \alpha/8)$ participation level in each epoch.

We must also consider the impact of the differential privacy guarantee we give on the statistic output by the Estimation Mechanism.
\begin{lemma}
Except with probability at most $1/12$, we have:
$$\left|\hat{a} - \sum_{x_i \in \textrm{AcceptedSet}_j}Q(x_i)\right| \leq \frac{\alpha}{4}\textrm{EpochSize}(\hat{j})$$
\end{lemma}
\begin{proof}
By the properties of the Laplace distribution, we have:
\begin{eqnarray*}
\Pr\left[\left|\hat{a} - \sum_{x_i \in \textrm{AcceptedSet}_j}Q(x_i)\right| \geq \frac{\alpha}{4}\textrm{EpochSize}(\hat{j})\right] &=& \exp\left(-\frac{100}{4}\log(\hat{j}+1)\right) \\
&=& \left(\frac{1}{\hat{j}+1}\right)^{100/4} \\
&<& \frac{1}{12}
\end{eqnarray*}
\end{proof}
We can now finish the proof. Except with probability $3\cdot\frac{1}{12} = \frac{1}{4}$, the conclusions of all of the previous 3 lemmas hold. We therefore have by the triangle inequality:

\begin{align*}
\left|\frac{\hat{a}}{\textrm{EpochSize}(\hat{j})} - \E_{\mathcal{D}}[Q(x)]\right| \leq& \left|\frac{\hat{a}}{\textrm{EpochSize}(\hat{j})} - \frac{\sum_{x_i \in \textrm{AcceptedSet}_j}Q(x_i)}{\textrm{EpochSize}(\hat{j})}\right|\\
&+ \left|\frac{\sum_{x_i \in \textrm{AcceptedSet}_j}Q(x_i)}{\textrm{EpochSize}(\hat{j})}- \frac{\sum_{x_i \in \textrm{Epoch}_j}Q(x_i)}{\textrm{EpochSize}(\hat{j})}\right|\\
&+ \left|\frac{1}{\textrm{EpochSize}(j)}\sum_{x_i \in  \textrm{Epoch}_j } Q(x_i) - \E_\mathcal{D}[Q(x)]\right|\\
\leq& \frac{\alpha}{4} + \frac{\alpha}{4} + \frac{\alpha}{4} \\
< &\alpha,
\end{align*}

which completes the proof.
\end{proof}

Note that we have made no attempt to optimize the constants in this section.

\subsection{Benchmark Comparison}
In this section we compare the cost of our mechanism to the cost of an
omniscient mechanism that is constrained to make envy-free offers and
achieve $\theta(\alpha)$-accuracy. For the purposes of the cost
comparison, in this section we assume that the individuals our
algorithm approaches have linear cost functions: $c_i(\epsilon) =
v_i\epsilon$ for some $v_i \in \mathbb{R}^+$.

We will use a result of Ghosh and Roth, translated into our setting.
\begin{theorem}[\cite{GR11}]
\label{thm:GR}
Let $0 < \alpha < 1$. Given any finite sample of size $n$, any differentially private mechanism that is $\alpha/4$ accurate must select a set of agents $H \subseteq [n]$ such that:
\begin{enumerate}
\item $\epsilon_i \geq \frac{1}{\alpha n}$ for all $i \in H$
\item $|H| \geq (1-\alpha)n$
\end{enumerate}
\end{theorem}

Let $v(\alpha)$ be the smallest value $v$ such that $\Pr_{x_i \sim
  \mathcal{D}} [v_i \leq v] \geq 1-\alpha$. In other words,
$(v(\alpha)\cdot 2\epsilon, \epsilon,\epsilon)$ is the cheapest
take-it-or-leave-it offer for $\epsilon$-units of privacy that in the
underlying population distribution would be accepted with probability
at least $1-\alpha$,
 It follows immediately that:
 \begin{observation}
   Any $(\alpha/32)$-accurate mechanism that makes the same
   take-it-or-leave-it offer to every individual $x_i \sim
   \mathcal{D}$ must in expectation pay in total at least
   $\Theta(\frac{v(\frac{\alpha}{8})}{\alpha})$. Note that
   because here we assume cost functions are linear, this quantity is
   fixed independent of the number of agents the mechanism draws from
   $\mathcal{D}$.
 \end{observation}

\begin{proof}
  Note that if the benchmark algorithm were to incur sample error,
  this would only strengthen our lower bound.


By Theorem \ref{thm:GR}, any $\alpha/32$-accurate algorithm must offer a high enough price to get a participation rate of at least $(1-\alpha/8)$ at a privacy level of at least $\epsilon = 32/(\alpha n)$. The cost for such a participation rate is $32/(\alpha\cdot n)\cdot v(\alpha/8)$ by the definition of $v$. Suppose the mechanism samples $n$
individuals. Then the mechanism must in expectation pay
$v(\alpha/8)\cdot 32 /(\alpha n)$ to
$n(1 - \alpha/8)$ individuals.
\end{proof}

We now bound the expected cost of our mechanism, and compare it to our benchmark cost, $\textrm{BenchMarkCost} = \Theta(\frac{v(\frac{\alpha}{8})}{\alpha})$

\begin{theorem}
The total expected cost of our mechanism is at most:
\begin{eqnarray*}
\E[\textrm{MechanismCost}] &=& O\left(\log \log \left(\alpha\cdot v(\alpha/8)\right)\cdot \textrm{BenchMarkCost} + \frac{1}{\alpha^2}\right) \\
&=& O\left(\log \log \left(\alpha^2\cdot\textrm{BenchMarkCost}\right)\cdot \textrm{BenchMarkCost} + \frac{1}{\alpha^2}\right)
\end{eqnarray*}
\end{theorem}
\begin{remark}
Note that the additive $1/\alpha^2$ term is necessary only in the case in which $v(\alpha/8) \leq (1+\eta)/\alpha$: i.e., only in the case in which the \emph{very first} offer will be accepted by a $1-\alpha/8$ fraction of players with high probability. In this case, we have started off offering too much money, right off the bat. An additive term is necessary, intuitively, because we cannot compete with the benchmark cost in the case in which the benchmark cost is arbitrarily small.\footnote{We thank Lisa Fleischer and Yu-Han Lyu for pointing out the need for the additive term.}
\end{remark}
\begin{proof}
 Let $j^*$ be the minimum epoch number such that the price offered is at least $v \left(\alpha/8\right)\epsilon_0 = \alpha v(\alpha/8)$: $p_{j^*} = (1+\eta)^{j^*} \geq v(\alpha/8)\cdot \alpha$. This is the first round at which the price offered is high enough to guarantee an expected rate of participation above $(1-\alpha/8)$. Note that if $j^* > 1$ we also have $p_{j^*} \leq (1+\eta)v(\alpha/8)\cdot \alpha$. Our proof will proceed by arguing that $\E[\textrm{MechanismCost}]$ is within a small constant factor of the cost incurred during epoch $j^*$.

We first argue that the total cost incurred during all epochs $j < j^*$ is comparable to the cost incurred at epoch $j^*$, which is at most: $(1+\eta)^{j^*}\cdot \textrm{EpochSize}(j^*)$. We will write Cost$(i)=p_i\cdot |\textrm{AcceptedSet}_i|$ to denote the total cost of epoch $i$. Therefore by the definition of $j^*$ we have in the case in which $j^* > 1$:
\begin{eqnarray*}
\textrm{Cost}(j^*) &\leq& (1+\eta)v(\alpha/8)\alpha\cdot\left(\frac{\log j^*+1}{\alpha^2}\right) \\
&=& \Theta\left(\frac{v(\alpha/8)\log j^*}{\alpha}\right) \\
&=& \Theta\left(\frac{v(\alpha/8)\log \log (\alpha \cdot v(\alpha/8))}{\alpha}\right) \\
&=& \Theta\left(\log \log (\alpha \cdot v(\alpha/8))\cdot \textrm{BenchMarkCost}\right)
\end{eqnarray*}
In the case in which $j^* = 1$, we have $\textrm{Cost}(j^*) = (1+\eta)\cdot \textrm{EpochSize}(1) = O(1/\alpha^2)$. Thus, in both cases we have:
$$\textrm{Cost}(j^*) \leq  O\left(\log \log (\alpha \cdot v(\alpha/8))\cdot \textrm{BenchMarkCost} + \frac{1}{\alpha^2}\right)$$

Therefore, the theorem will be proven if we can argue that $\E[\textrm{MechanismCost}] = O(\textrm{Cost}(j^*))$. The remainder of the proof will establish this claim.

\begin{theorem}
$\E[\textrm{MechanismCost}] = O(\textrm{Cost}(j^*))$ whenever $\eta$ is a constant such that $c_1 < \eta < \frac{3}{17}-c_2$ where $c_1$ and $c_2$ are constants bounded away from $0$.
\end{theorem}
\begin{proof}
We first argue that the contribution to the cost of epochs $j < j^*$
is small.
\begin{lemma}
$$\sum_{i=1}^{j^*-1}\textrm{Cost}(i) \leq \frac{(1+\eta)^{j^*}}{\eta}\cdot \textrm{EpochSize}(j^*)$$
\end{lemma}
\begin{proof}
\begin{eqnarray*}
\sum_{i=1}^{j^*-1}\textrm{Cost}(i) &=& \sum_{i=1}^{j^*-1}(1+\eta)^i|\textrm{AcceptedSet}_i| \\
&\leq& \textrm{EpochSize}(j^*)\sum_{i=1}^{j^*-1}(1+\eta)^i \\
&<& \frac{(1+\eta)^{j^*}}{\eta}\cdot \textrm{EpochSize}(j^*)
\end{eqnarray*}
\end{proof}

We next argue that the contribution to the expected cost of the algorithm of epochs $j >  j^*$ is small.
\begin{lemma}
For any epoch $j > j^*$, the probability that the algorithm reaches epoch $t$ before halting is at most $(17/20)^{t-j^*}$.
\end{lemma}
\begin{proof}
First recall that by definition of $j^*$, we have for any $j > j^*$
$$\Pr_{x_i \sim \mathcal{D}}[x_i \textrm{ accepts } (p_{j}, \epsilon_0, \epsilon_0)] \geq 1-\alpha/8$$
We therefore have:
$$\Pr[|\textrm{AcceptedSet}_j| < (1-\alpha/8)\cdot \textrm{EpochSize}(j) - 1] \leq \frac{1}{2}$$
Note that round $j$ will be the final round if
$\textrm{EpochSize}(j)+\nu_{j} \geq
(1-\alpha/8)\textrm{EpochSize}(j)$.

Conditioned on the event $|\textrm{AcceptedSet}_j| \geq (1-\alpha/8)\cdot \textrm{EpochSize}(j)-1$, we have by the properties of the Laplace distribution:
$$\Pr[\textrm{EpochSize}(j)+\nu_{j} \geq (1-\alpha/8)\textrm{EpochSize}(j)] \geq \frac{1}{2}\cdot \exp(-\epsilon_0) \geq \frac{3}{10}$$
Therefore, each round $j > j^*$ is the final round with probability at least $3/20$. Since each of these events is independent, the probability that the algorithm does not halt before epoch $t$ is at most $(17/20)^{t-j^*}$ as desired.
\end{proof}

Write $H_j$ for the event that the mechanism does not halt before round $j$. The expected cost of the mechanism is then at most:

\begin{align*}
\sum_{i=1}^{j^*}\textrm{Cost}(i) +& \sum_{j = j^*+1}^\infty (1+\eta)^j\cdot \textrm{EpochSize}(j)\cdot \Pr[H_j]\\
&\leq
 \frac{(1+\eta)^{j^*}}{\eta}\cdot \textrm{EpochSize}(j^*) + \sum_{j = j^*+1}^\infty (1+\eta)^j\cdot \textrm{EpochSize}(j)\cdot \Pr[H_j] \\
&\leq
\frac{(1+\eta)^{j^*}}{\eta}\cdot \textrm{EpochSize}(j^*) + \sum_{j = j^*+1}^\infty (1+\eta)^j\cdot \textrm{EpochSize}(j)\cdot(17/20)^{j-j^*} \\
&= O\left(\frac{(1+\eta)^{j^*}}{\eta}\cdot \textrm{EpochSize}(j^*)\right)
\end{align*}
whenever there exist constants $c_1, c_2$ bounded away from $0$ such that $c_1 < \eta < \frac{3}{17}-c_2$.
\end{proof}
Therefore, we have shown that
$$\E[\textrm{MechanismCost}] = O(\textrm{Cost}(j^*)) = O\left(\log \log (\alpha\cdot v(\alpha/8))\cdot \textrm{BenchMarkCost} + \frac{1}{\alpha^2}\right)$$
which completes the proof.
\end{proof}

\section{Discussion}
In this paper, we have proposed a method for accurately estimating a
statistic from a population that experiences cost as a function of
their privacy loss. The statistics we consider here take the form of
the expectation of some predicate over the population. We leave to
future work the consideration of other, nonlinear, statistics. We have
circumvented the impossibility result of \cite{GR11} by using a
mechanism empowered with the ability to approach individuals and make
them take-it-or-leave-it offers (instead of relying on a direct
revelation mechanism), and by relaxing the measure by which
individuals experience privacy loss. Moving away from direct
revelation mechanisms seems to us to be inevitable: if costs for
privacy can be correlated with private data, then merely asking for
individuals to report their costs is inevitably disclosive, for any
reasonable measure of privacy. On the other hand, we do not claim that
the model we use for how individuals experience cost as a function of
privacy is ``the'' right one. Nevertheless, we have argued that some
relaxation away from individuals experiencing privacy cost entirely as
a function of the differential privacy parameter of the entire
mechanism is inevitable (as made particularly clear in the setting of
take-it-or-leave-it offers, in which individuals in this model would
accept arbitrarily low offers). In particular, we believe that the
style of survey mechanism presented in this paper, in which the
mechanism may approach individuals with take-it-or-leave-it offers, is
realistic, and any reasonable model for how individuals value their
privacy should predict reasonable behavior in the face of such a
mechanism.

\bibliographystyle{alpha}
 \bibliography{buyingprivacy}
\end{document}